\documentclass[10pt,journal,compsoc]{IEEEtran}

\usepackage{graphicx,verbatim}
\usepackage{amssymb,amsthm,amsmath}
\usepackage{epstopdf}
\usepackage[usenames]{color}

\theoremstyle{plain}
\newtheorem{thm}{Theorem}[section]

\newtheorem{cor}[thm]{Corollary}

\theoremstyle{definition}

\newtheorem{ex}[thm]{Example}

\theoremstyle{remark}

\newcommand{\PP}{\mathbb P}
\newcommand{\EE}{\mathbb E}

\newcommand{\tc}{\text{:}} % tc=tight colon, for use in math mode in Newick trees

\ifCLASSOPTIONcompsoc
  \usepackage[nocompress]{cite}
\else
  \usepackage{cite}
\fi

\begin{document}

\title{Species tree inference from gene splits by Unrooted STAR methods}

\author{Elizabeth~S.~Allman, James~H.~Degnan, and John~A. ~Rhodes
\IEEEcompsocitemizethanks{
\IEEEcompsocthanksitem E.S. Allman and J.A. Rhodes are with the Department
of Mathematics and Statistics, University of Alaska Fairbanks, Fairbanks, AK 99775.\protect\\
E-mail: e.allman@alaska.edu, j.rhodes@alaska.edu
\IEEEcompsocthanksitem J.H. Degnan is with the Department
of Mathematics and Statistics, The University of New Mexico, Albuquerque, NM  87131
E-mail: jamdeg@unm.edu 
}% <-this % stops an unwanted space
%\thanks{Manuscript received ?????.}
}

%{Shell \MakeLowercase{\textit{et al.}}: Bare Demo of IEEEtran.cls for Computer Society Journals}

\IEEEtitleabstractindextext{%
\begin{abstract}
The  $\text{NJ}_{st}$ method was proposed by Liu and Yu to infer a species tree topology from unrooted topological gene trees. While its statistical consistency under the  multispecies coalescent model was established only for a 4-taxon tree, simulations demonstrated its good performance on gene trees inferred from sequences for many taxa. Here we prove the statistical consistency of the method for an arbitrarily large species tree. Our approach connects $\text{NJ}_{st}$ to a generalization of the STAR method of Liu, Pearl and Edwards, and a previous theoretical analysis of it. We further show $\text{NJ}_{st}$
utilizes only the distribution of splits in the gene trees, and not their individual topologies. Finally, we discuss how multiple samples per taxon per gene should be handled for  statistical consistency.
\end{abstract} 
 \begin{IEEEkeywords}
 coalescent model, STAR algorithm, $\text{NJ}_{st}$, species tree.
 \end{IEEEkeywords}}

\maketitle

\IEEEraisesectionheading{\section{Introduction}}
\IEEEPARstart{W}{ith} the growing feasibility of building large multi-locus data sets of genetic sequences, questions of how to best infer ancestral relationships between taxa have increasingly been viewed in the light of the multispecies coalescent model. This model describes the formation of \emph{gene trees} (or genealogies) relating orthologous loci within a \emph{species tree} composed of populations. It thus brings into phylogenetics an important model of population genetics, in order to capture the phenomenon of incomplete lineage sorting. While the coalescent framework still omits the possibility of non-tree-like relationships due to hybridization or lateral gene transfer, it allows incongruence across gene trees to be used to more accurately infer species trees in situations where incomplete lineage sorting is the dominant cause.

In principle it is straightforward to combine standard models of sequence evolution with the multispecies coalescent for inference of species trees under either a maximum likelihood or Bayesian framework. In practice, though, this is both computationally intensive and requires some additional assumptions --- most importantly, a means of relating the coalescent and mutation time scales --- that may or may not be reasonable. Such assumptions are not always highlighted in data analysis, even though they may include 1) a molecular clock operating for each gene tree, 2) constant population sizes over each branch of the species tree, and 3) a common mutation rate across gene trees, or 
variants of these. It is also not clear to what extent existing software implementations have been applied to simulated data violating such assumptions, in order to understand their robustness. Finally, even accepting these assumptions, analyzing a many-gene many-taxon data set can be computationally infeasible using a standard approach.

Some inference approaches simplify the problem by first inferring individual gene trees by established phylogenetic methods, and then using these to infer a species tree. From the gene trees, one might use metric information, or only topologies, with or without a root. If one views the gene tree topologies as more robustly inferable than metric edge lengths, then two methods, STAR \cite{Liu2009}, and $\text{NJ}_{st}$ \cite{Liu2011}, are especially attractive. By not using any metric information from the gene trees, they elegantly circumvent issues of how one should relate the coalescent and mutational time scales.
They both encode gene tree topologies through special distance matrices, in what one might call a remetrization step, with STAR requiring rooted trees, and $\text{NJ}_{st}$ unrooted ones. The average of these matrices is then used as input to a standard metric tree-building algorithm to recover the species tree topology. (Though one obtains edges lengths as well as part of the process, their relationship to the true lengths on the species tree is not currently known.) All computations are both simple and fast, and accuracy on large datasets is competitive with the best current methods, as shown by the recent implementation and extension of $\text{NJ}_{st}$ in the software ASTRID \cite{ASTRIDWarnow}.

In \cite{Liu2009} and  \cite{Liu2011} arguments were given establishing the statistical consistency of STAR and $\text{NJ}_{st}$ for 4-taxon trees. In \cite{adr2013} a  rigorous proof of consistency was given for STAR and variants of it on arbitrarily large trees, along with a theoretical exploration of how the algorithm actually only required the distribution of clades on the gene trees. This recasts STAR as a clade consensus method attuned specifically to the multispecies coalescent model.

Here we obtain similar theoretical results for $\text{NJ}_{st}$. We first prove its statistical consistency under the multispecies coalescent model on arbitrary trees in Theorem \ref{thm:main}. Our proof is built on relating $\text{NJ}_{st}$ to a generalized STAR method as introduced in \cite{adr2013}, and deducing our results from those on STAR.
In Theorem \ref{thm:splitUSTAR} we show the method uses only information in the distribution of splits on gene trees, and not the more detailed information of the gene trees themselves. Thus we view it as a split consensus method designed specifically for inference of species trees under the multispecies coalescent model. 
In Section \ref{sec:multi} we then discuss how one might apply the method to data that involves multiple samples from each taxon. We show the approach suggested by \cite{Liu2011} for such data can be problematic even in a simple case, but then give an alternative which is statistically consistent under certain sampling schemes. 

Finally, we suggest a rechristening of $\text{NJ}_{st}$ as USTAR/NJ, for ``Unrooted STAR with Neighbor Joining.'' This both emphasizes the close relationship of the two methods, and emphasizes that one might perform the method with tree selection algorithms other than Neighbor Joining. Any statistically consistent method for selecting a metric tree from possibly non-ultrametric distance tables could be used in its place. For instance, USTAR/BIONJ \cite{BIONJ} uses a different purely algorithmic tree building method,  while USTAR/FastME \cite{FastME} performs a hueristic search to optimize the balanced minimum evolution criterion to select a tree.  
Indeed, the ASTRID software already allows one to apply such methods and \cite{ASTRIDWarnow} compares their performance.

\section{Notation and Terminology}
Let $\mathcal X$ be a finite set of $n$ taxa, which we denote by lower case letters $a,b,c,\dots$. For any specific gene, we denote a single sample from each taxon by the corresponding upper case letter $A,B,C,\dots$, with $\mathcal X_g$ the full set of such gene samples. If $\mathcal A\subseteq \mathcal X$ is a subset of taxa, the corresponding subset of genes is $\mathcal A_g\subseteq \mathcal X_g$. For example $\{a,b,c\}_g=\{A,B,C\}$.

By a species tree $\sigma=(\psi,\lambda)$ on $\mathcal X$ we mean a rooted topological tree with leaves bijectively labelled by $\mathcal X$, together with an assignment of edge weights $\lambda$ to its internal edges. These edge weights are specified in coalescent units, so that the multispecies coalescent model on $\sigma$ leads to a probability distribution on gene trees with leaves labelled by $\mathcal X_g$. (For a more precise definition of the multispecies coalescent as we use it, we direct the reader to \cite{adr2011a}.) The gene trees here are metric rooted
trees, though this probability distribution, by marginalization, also leads to ones on metric or topological, rooted or unrooted, gene trees. We denote rooted topological gene trees by $T^r$ and unrooted topological gene trees by $T$. The probability of an unrooted topological gene tree $T$ under the multipspecies coalescent on $\sigma$ is denoted $\PP_\sigma(T)$.

A metric tree is called \emph{binary} if the underlying topological tree is binary and all internal edge lengths are positive.
\smallskip

A \emph{split} of a set of taxa $\mathcal X$ is a bipartition $\mathcal A |\mathcal B$ of $\mathcal X$ in which neither $\mathcal A$ nor $\mathcal B$ is empty. Note  $\mathcal A |\mathcal B$  is the same split as  $\mathcal B |\mathcal A$. 
If $\sigma=(\psi,\lambda)$  is a species tree on $\mathcal X$ then
a \emph{split on $\sigma$} is a split of $\mathcal X$
formed by deleting a single edge of $\psi$ and grouping taxa according to the connected components of the resulting graph. 
We similarly define splits of $\mathcal X_g$, and splits of $\mathcal X_g$ on  a specific gene tree.

\section{USTAR methods}
Given an unrooted topological gene tree $T$ on $\mathcal X_g$, we may  metrize it by giving all edges length 1. The distance $D_T(A,B)$ between any two gene samples $A,B$ on $T$ is then the number of edges in the path connecting them, \emph{i.e.}, the graph-theoretic distance. Fixing an ordering of the taxa, it is convenient to think of $D_T$ as an $n\times n$ matrix. In essence, we have simply encoded the topology of $T$ by the numerical matrix $D_T$.

In \cite{Liu2011}, the internodal distance, i.e., the number of nodes on the path in the tree between two taxa, is used to define a similar distance table. The graph-theoretic distance between taxa is always one more than the internodal distance, and it is straightforward to check that this difference between them has no essential impact on anything we do in this paper. We use the graph-theoretic distance here for its simple interpretation in terms of assigning edge lengths of 1, and its more direct connection to the notion of splits on the tree.

For a probability distribution $\mu$ on unrooted gene trees on $\mathcal X_g$, the expected value 
$$D:=\EE_\mu(D_T)=\sum_T \mu(T) D_T$$
defines a dissimilarity function on $\mathcal X_g$. Identifying $\mathcal X$ with $\mathcal X_g$, we call this the \emph{USTAR 
dissimilarity on $\mathcal X$} with respect to $\mu$.  For an empirically-obtained collection of gene trees, this dissimilarity is just the mean of the matrices $D_T$ for trees in the sample.

In this paper, we focus on the particular choice $\mu=\PP_\sigma$, i.e., we use the probability of unrooted gene trees arising from the multispecies coalescent on a specific species tree $\sigma$, or an empirical distribution describing a sample from this theoretical one.

From the USTAR dissimilarity $D$ obtained from a gene tree distribution, one can construct or choose a tree on $\mathcal X$, using any of a variety of well-known methods --- \emph{e.g.}, UPGMA, Neighbor Joining, BIONJ, Balanced Minimum Evolution, etc. Discarding any edge lengths that might have been produced in the course of applying the tree selection method, yields a topological tree on $\mathcal X$. Thus we have a family of methods whose input is a theoretical or empirical distribution of unrooted topological gene trees, and output is a single unrooted topological tree on the taxa.
In particular, USTAR/NJ is the method obtained when Neighbor Joining is used, and coincides with $\text{NJ}_{st}$. The output of such a method can be viewed as an inferred species tree, but the validity of this view hinges on the question of whether the method is statistically consistent.

\medskip

USTAR methods can be helpfully viewed as related to generalized STAR methods developed in \cite{adr2013}, building on \cite{Liu2009}. STAR methods of inferring a species tree from \emph{rooted} gene trees similarly involve metrizing the gene trees and averaging the resulting pairwise distance matrices over a gene tree distribution. However the metrization is done as follows: For $n$ taxa, first choose a non-increasing sequence of node numbers $a_0\ge a_1\ge a_2\ge \dots\ge a_{n-2}\ge 0$, with at least one of these inequalities strict. Assign $a_0$ to the root, $a_1$ to its non-leaf children, $a_2$ to their non-leaf children, etc. Then interpret the assigned numbers as distances from the leaves in an ultrametric tree.

For the particular case of node numbers $n-3/2,n-2,n-3,n-4,\dots$, the generalized STAR metrization has the effect of giving length 1 to all internal edges of the rooted gene tree, except those incident to the root. However, if suppressing the root leads to a new internal edge in the unrooted version, the total length of that edge is 1. Thus after suppressing the root, all internal edges of the gene tree are given the same length as they would be by USTAR. However, lengths of pendant edges are different, as they are all 1 under USTAR and they vary to achieve ultrametricity under STAR.

\begin{ex} 

\begin{figure}
\includegraphics[width=9cm]{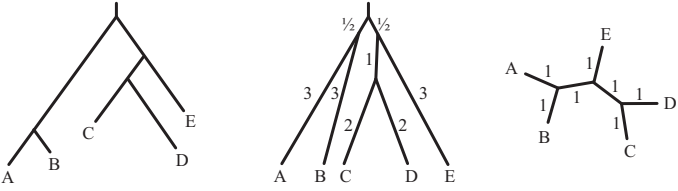}
\caption{ A gene tree ((A,B),((C,D),E)) (left) and its metrizations for the generalized STAR method discussed in the text (center), and USTAR (right).}\label{fig:example}
\end{figure}

Consider the 5-taxon gene tree $T^r=((A,B),((C,D),E))$ shown in Figure \ref{fig:example}. Viewing $T$ as an unrooted tree, with taxa ordered alphabetically, USTAR leads to the distance matrix
$$D_T=\begin{pmatrix} 
0&2&4&4&3\\2&0&4&4&3\\4&4&0&2&3\\4&4&2&0&3\\3&3&3&3&0\end{pmatrix}.$$
Separating the contributions from internal and pendant edges, we can write
\begin{align*}
D_T&=D_T^{ui} +D^{up}_T\\
& =
\begin{pmatrix} 
0&0&2&2&1\\0&0&2&2&1\\2&2&0&0&1\\2&2&0&0&1\\1&1&1&1&0
\end{pmatrix}+\begin{pmatrix} 
0&2&2&2&2\\2&0&2&2&2\\2&2&0&2&2\\2&2&2&0&2\\2&2&2&2&0
\end{pmatrix}.
\end{align*}
 Here `$ui$' and `$up$' refer to the `unrooted internal' end `unrooted pendant' edge contributions.

Viewing $T$ as a rooted tree, $T^r$, STAR with the node numbering given above, leads to the distance matrix
$$D^r_{T^r}=\begin{pmatrix} 
0&6&7&7&7\\6&0&7&7&7\\7&7&0&4&6\\7&7&4&0&6\\7&7&6&6&0
\end{pmatrix}.$$
Again separating the contributions from internal and pendant edges of the unrooted tree, we  have
\begin{align*}
D^r_{T^r}&=D_T^{ui} +D^{rp}_{T^r}\\
&= \begin{pmatrix} 
0&0&2&2&1\\0&0&2&2&1\\2&2&0&0&1\\2&2&0&0&1\\1&1&1&1&0
\end{pmatrix}+\begin{pmatrix} 
0&6&5&5&6\\6&0&5&5&6\\5&5&0&4&5\\5&5&4&0&5\\6&6&5&5&0
\end{pmatrix},
\end{align*}
where `$rp$' refers to the `rooted pendant' edge contributions. 
For a general tree, the rooted pendant edge contributions may include some that arise from an internal edge incident to the root that becomes part of a pendant edge when the root is suppressed (such as when there is a single outgroup on the tree).

Note that the same contributions appears from the internal edges of the unrooted tree in both the USTAR and STAR distance matrices.
Our analysis of USTAR in the proof of Theorem \ref{thm:Udist} below will be based in the fact that, for the particular STAR numbering scheme where branches incident to the root have length 1/2 and all other internal branches have length 1, the distance matrices differ only in contributions from pendant edges.
\end{ex}

\section{Statistical consistency}

Our goal in this section is to prove the following:

\begin{thm} \label{thm:main}Let M denote any method of obtaining an unrooted topological tree from a dissimilarity function satisfying
\begin{enumerate}
\item $M$ applied to a tree metric returns the unique tree fitting it, and\label{cond:fit}
\item $M$ is continuous at tree metrics arising from binary trees. \label{cond:cont}
\end{enumerate}
Let
$\sigma=(\psi,\lambda)$ be a binary species tree on $\mathcal X$. Then USTAR/M is a statistically consistent method of inference of the unrooted topology of $\psi$ from gene trees under the multispecies coalescent model on $\sigma$. \end{thm}

Informally, the continuity required of the method $M$ in condition \eqref{cond:cont}  means that when $M$ is applied to a sufficiently small perturbation of a binary tree metric, it returns the correct tree topology, and edge lengths close to those underlying the tree metric.  As NJ is known to satisfy conditions \eqref{cond:fit}  and \eqref{cond:cont}, we see that in particular USTAR/NJ is consistent. Since UPGMA does not, in general, satisfy condition \eqref{cond:fit} for non-ultrametric trees, the theorem does not apply to USTAR/UPGMA. 

Theorem \ref{thm:main} is a consequence of the following.
\begin{thm} \label{thm:Udist}
The USTAR dissimilarity on $\mathcal X$ with respect to the probability distribution on unrooted topological gene trees arising from multispecies coalescent model on $\sigma=(\psi,\lambda)$, $D=\EE_\sigma(D_T)$, exactly fits the unrooted species tree topology of $\psi$.
\end{thm}

\begin{proof}
Let $T^r$ denote a rooted gene tree topology. Consider the generalized STAR number scheme for rooted gene trees with node numbering sequence $n-3/2,n-2,n-3,n-4,\dots$. As discussed previously, when the root is suppressed on the STAR remetrized rooted gene tree $T^r$, all internal edges on the resulting unrooted tree have length 1. Using this node numbering scheme, let $D^r_{T^r}$ denote the STAR distance matrix for a remetrized rooted tree $T^r$, and $\EE_\sigma(D^r_{T^r})$ its expected value under the distribution on rooted topological gene trees arising from the coalescent.

We now relate the STAR dissimilarities $D^r=\EE_\sigma(D^r_{T^r})$ to those of USTAR, $D=\EE_\sigma(D_T)$. Since both the rooted and unrooted schemes give each internal edge length 1 in the unrooted gene tree topology we can write
\begin{equation}
D_T=D^{ui}_T+D^{up}_T, \ \ \ D^r_{T^r}=D^{ui}_T+D^{rp}_{T^r}\label{eq:Dint}
\end{equation}
where $D^{ui}_T$ contains the contributions to distances from internal edges of the unrooted tree, $D^{up}_T$ contains contributions from pendant edges of the unrooted scheme, and  $D^{rp}_{T^r}$ contains contributions from (unrooted tree) pendant edges in the rooted scheme.
Equations \eqref{eq:Dint} thus imply
\begin{equation}D_T= D^r_{T^r} +D^{up}_T- D^{rp}_{T^r}.\label{eq:Ddiff}
\end{equation}

Now the matrix $D^{up}_T$ is independent of $T$ and has a simple structure; all diagonal entries are 0, and all off-diagonal entries are $1+1=2$. 
The matrix 
$D^{rp}_T$, however,  does depend on $T^r$. While it also has 0 in every diagonal entry, the off-diagonal entry in row $x$, column $y$ is $w_x+w_y$, where $w_x, w_y$ are the lengths assigned to the pendant edges to taxa $x,y$ after the root is suppressed on the remetrized ultrametric $T^r$.

 Passing to expected values, we have  from equation \eqref{eq:Ddiff} that
\begin{equation}
D= D^r +\EE_\sigma(D^{up}_T) - \EE_\sigma( D^{rp}_{T^r}).\label{eq:ED}
\end{equation}
By Theorem 3.2 of \cite{adr2013}, $D^r$ exactly fits the topology of the rooted species tree (ultrametrically), and hence for each choice of 4 taxa, with some permutation of their labels the 4-point condition
\begin{align}D^r(a,c)+D^r(b,d)&=D^r(a,d)+D^r(b,c)\notag\\
&\ge D^r(a,b)+D^r(c,d)\label{eq:Dr4}\end{align}
holds.
Now in the case that $a,b,c,d$ are all distinct, this implies
\begin{align}D(a,c)+D(b,d)&=D(a,d)+D(b,c)\notag\\
&\ge D(a,b)+D(c,d)\label{eq:Dr4b},\end{align}
since
by equation \eqref{eq:ED}, we have only added $2-\EE(w_a+w_b+w_c+w_d)$ to the three sums in  \eqref{eq:Dr4} to obtain  \eqref{eq:Dr4b}.

If at most 3 of the taxa in the 4-point condition are distinct this last argument is not valid. However, if, say, $c=d$, the 4-point condition we need to establish degenerates to
$$D(a,c)+D(b,c)\ge D(a,b).$$
That this holds follows from the fact the corresponding inequality
holds for every tree metric, and in particular for each USTAR remetrization $D_T$, and hence for the expected value as well.

Thus the four point condition holds for $D$ for every set of 4 taxa, and it yields the same unrooted quartet topology as does $D^r$. Thus by standard results in \cite{SempleSteel} $D$ exactly fits the same unrooted tree topology as $D^r$, which is that of the species tree.
\end{proof}

\begin{proof}[Proof of Theorem \ref{thm:main}]
As the size of a sample of gene trees from the multispecies coalescent model on $\sigma$ increases, the empirical distribution of unrooted gene tree topologies approaches the exact one with probability 1, and thus the empirical USTAR dissimilarity approaches the theoretical one, $D$.
Since Theorem  \ref{thm:Udist} and condition \eqref{cond:cont} ensures the method $M$ returns the correct tree when applied to $D$,  condition \eqref{cond:fit} then implies with probability 1 USTAR/$M$ returns the correct unrooted species tree topology as the sample size increases to infinity. \end{proof}

\section{USTAR and splits} 

Here we establish a relationship between the USTAR expected distance matrix and split probabilities, analogous to that given in \cite{adr2013} for STAR expected distances and clade probabilities.

As a consequence of this relationship, it is natural to view USTAR methods as a type of \emph{split consenus method}. Specifically, USTAR methods use only information on probabilities of splits on gene trees, and not the finer information of the gene tree topologies themselves.

The fact that USTAR uses only split frequencies, yet can produce statistically consistent inference for the coalescent model is notable, as other split methods lack this feature. For instance
greedy consensus \cite{bryant2003classification} accepts splits in order of decreasing frequency, if they are compatible with previously accepted splits. 
Greedy consensus on clades has been proven inconsistent \cite{DegnanEtAl09}, though STAR can be viewed as a consistent clade consensus method \cite{adr2013}.  The arguments in  \cite{DegnanEtAl09} can be modified to give a similar result for greedy consensus on splits, with signs of inconsistent behavior also observed in simulations \cite{mirarab2014evaluating}.  For consistency, a consensus method must be attuned to the model of tree variation, with USTAR and STAR being appropriate for the coalescent.

\medskip

Given any two leaves $A,B$ of a gene tree $T$, let $S_T^{A,B}$ denote the set of
splits of $T$ in which $A$ and $B$ are separated (\emph{i.e.}, in different bipartition sets).  
Elements of $S_T^{A,B}$ correspond to the
edges of $T$ lying on the path from the $A$ to $B$, 
so
\begin{equation} D_T(A,B) =|S_T^{A,B}|.
\label{eq:cd}
\end{equation}
This means on an individual gene tree the distances used
in USTAR are simply counts of `separating' splits, with gene samples being judged
further apart when there are more splits on $T$ which separate them.
Thus graph-theoretic distance might also be called `split separation distance.' 
\medskip

Now for any distribution $\mu$ of 
gene trees, if $\mathcal A|\mathcal B$ is a
split of $\mathcal X$, and $\PP_\mu(\mathcal A|\mathcal B)$ denotes the
probability of the event that an observed gene tree displays split $\mathcal A_g|\mathcal B_g$, then
$$\PP_\mu(\mathcal A|\mathcal B)=
 \sum_{T \text{ displaying } \mathcal A_g|\mathcal B_g} \PP_\mu (T).$$

\begin{thm}\label{thm:splitUSTAR}
For any distribution $\mu$ of gene trees, the collection of split probabilities $\{\PP_\mu( \mathcal A|\mathcal B)\}$
determines $\EE_\mu(D_T)$.
\end{thm}

\begin{proof}   Define indicator functions
$$
I_{ \mathcal A|\mathcal B }(T) = \begin{cases} 1 &\text{if $T$ displays $ \mathcal A_g|\mathcal B_g$,}\\ 0 &\text{otherwise,} \end{cases}
$$
and
$$
J_{A,B} ( \mathcal A|\mathcal B) = \begin{cases} 1 &\text{if $A,B$ separated in  $\mathcal A_g|\mathcal B_g$, }\\ 0 &\text{otherwise.} \end{cases}
$$

\medskip

Then using equation \eqref{eq:cd},
\begin{align}
\EE_\mu(D_T(A,B)) &= \sum_{T} \PP_\mu (T) \, D_T(A,B)\notag\\
&= \sum_{T} \PP_\mu (T) \,\left |S_T^{A,B}\right |\notag\\
&= \sum_{T} \PP_\mu (T) \left(\sum_{\text{splits} \atop  \mathcal A|\mathcal B }  I_{\mathcal A|\mathcal B }(T) J_{A,B} (\mathcal A|\mathcal B ) \right)\notag\\
&=  \sum_{\text{splits} \atop  \mathcal A|\mathcal B } \left( \sum_{T} \PP_\mu (T) \, I_{\mathcal A|\mathcal B}(T) \right) J_{A,B} (\mathcal A|\mathcal B) \notag\\
&= \sum_{\text{splits} \atop  \mathcal A|\mathcal B } \PP_\mu (\mathcal A|\mathcal B )  J_{A,B} (\mathcal A|\mathcal B ),\label{eq:cprob}
\end{align}
so the USTAR dissimilarity is computable from split probabilities. 
\end{proof}

Of course the distribution $\mu$ we have in mind here is either the one arising from the multispecies coalsecent model, or an empirical one from a sample from that model.

From Theorems \ref{thm:splitUSTAR} and \ref{thm:Udist} we immediately obtain the following:

\begin{cor}\label{cor:SplitID}
  The unrooted species tree topology  is identifiable from split
  probabilities under the multispecies coalescent.
\end{cor}

It is known \cite{adr2011a} that the \emph{rooted} species tree topology is identifiable from the distribution of \emph{unrooted} gene tree topologies. It is also known that the rooted species tree topology is identifiable from clade probabilities.  Thus a natural question is whether the split probabilities, the unrooted analogues of clade probabilities, can further identify the root on the species tree. Though our investigation here does not seem to shed light on this, we plan to address it in another work.

\section{USTAR with multiple samples per taxon}\label{sec:multi}

When $\text{NJ}_{st}$ was introduced in \cite{Liu2011}, a suggestion was given for how one might deal with gene trees relating multiple lineages sampled from each taxon. 
For a collection $\mathcal T$ of gene trees, if $T\in \mathcal T$ relates $m_{a}(T)$ lineages sampled from taxon $a$ and $m_{b}(T)$ lineages from taxon $b$, then intertaxon distances were defined (up to an additive constant) as an average 
\begin{equation}
D(a,b)=\frac{\displaystyle \sum_{T\in \mathcal T}\sum_{1\le i\le m_{a}(T)\atop 1\le j \le m_{b}(T)} D_T(A_{i},B_{j})}{\displaystyle \sum_{T\in \mathcal T}m_{a}(T)m_{b}(T) },\label{eq:Liumulti}
\end{equation}
where $D_T(A_{i},B_{j})$ denotes the graph theoretic distance on tree $T$ between the $i$th sample of gene $A$ and the $j$th of gene $B$.
Unfortunately, this approach is not statistically consistent. In fact, as the size of the sample  of gene trees is increased, the probability of inferring the correct species tree can approach 0. After demonstrating this, we propose a different method of handling multiple samples per taxon, one that is statistically consistent.

\medskip

To investigate the behavior of formula \eqref{eq:Liumulti}, consider the species tree 
$$((a,b),(c,d)),$$ 
with all branch lengths long enough that incomplete
lineage sorting between different taxa is vanishingly rare. Sample lineages for a large number of genes as follows: For 50\% of the genes,  sample 3 lineages in each of taxa $a,b$ and 1 lineage in each of taxa $c,d$. In the other 50\% of genes,  sample 1 lineage in taxa $a,b$ and 3 lineages in taxa $c,d$. For sufficiently long branch lengths on the species tree, the coalescent model gives that the sampled genes trees will be approximately equally of topologies $$( ( ((A,A),A), ((B,B),B) ),(C,D))$$  and  $$((A,B),( ((C,C),C), ((D,D),D) ) ),$$
with an arbitrarily small fraction of gene trees of other topologies. 
For the first of these gene tree topologies, after unrooting and assigning all edges the length 1, the different interlineage USTAR distances are
\begin{align*}
D_T(A_i,B_j)&= 6,6,6,6,5,5,5,5,4;\\
D_T(x,y)&=5,5,4, \text{for $x=A_i,B_j$, $y=C,D$};\\
D_T(C,D)&=2.
\end{align*} For the second tree, the same distances arise, but with the roles of $A,B$ interchanged with $C,D$.
Then formula \eqref{eq:Liumulti} gives intertaxon distances arbitrarily close to
\begin{align*}
D&(a,b)=D(c,d)\\
&= \frac{(.5)(6+6+6+6+5+5+5+5+4)+.5(2)}{.5(9)+.5(1)}=5\\
D&(x,y)=\frac{ .5(5+5+4)+.5(5+5+4)}{.5(3)+.5(3)}=\frac {14}3\\
&\hskip 5cm\text{ for $x=a,b$, $y=c,d$}.
\end{align*}
These intertaxon distances do not fit any unrooted topological tree, as they do not satisfy the four-point condition \cite{SempleSteel}. In fact, selection of a tree topology by applying (part of) the four-point condition requires computing
\begin{align*}
D(a,b)+D(c,d)&=5+5=10,\\
D(a,c)+D(b,d)&=\frac {14}3+\frac {14}3 =\frac{28}3 \\
D(a,d)+D(b,c)&=\frac {14}3+\frac {14}3 =\frac{28}3
\end{align*} and choosing the smallest to determine the cherries of the tree. Here the smallest is a tie, yielding the  two incorrect topologies, $((a,c),(b,d))$ and $((a,d),(b,c))$. Neighbor Joining, which is built upon this selection criterion, would choose either of the incorrect topologies with equal probability, and then go onto compute positive lengths for the edges, obtaining either of the unrooted metric trees $((a\tc 2.333,c\tc 2.333)\tc 0.167,b\tc 2.333, d\tc 2.333)$ or $((a\tc 2.333,d\tc 2.333)\tc 0.167,b\tc 2.333, c\tc 2.333).$ 

Finite length edges on the species tree will only produce intertaxon distances arbitrarily close to those in the calculations above, with probability approaching 1 as the number of gene trees increases. However, continuity of the Neighbor Joining algorithm at these distances implies that the output of Neighbor Joining will be the wrong topology with probability approaching 1.

\medskip

 A different approach to averaging than the one used in formula \eqref{eq:Liumulti} can however lead to statistically consistent inference of the species tree.

First, suppose multiple samples are drawn from taxa in exactly the same number for each gene. That is, there are integers $m_x\ge1$ so that each gene tree has $m_x$ leaves  $X_1, X_2,\dots ,X_{m_x}$ for each $x\in \mathcal X$, for a total of 
$\sum_{x\in \mathcal X} m_x$ leaves. We will refer to a specific choice of the  numbers $(m_x)_{x\in \mathcal X}$ as a \emph{multisample scheme}.

For a single fixed multisample scheme, the results of previous sections apply if we replace the species tree by one where $m_x$ edges are attached to the leaf formerly labeled $x$ with the new leaves labeled $x_1, x_2,\dots, x_{m_x}$. (This is called the \emph{extended species tree} in \cite{adr2011b}.) While this tree is not binary, one can consider binary perturbations of it, and use continuity to conclude that
the expected USTAR dissimilarity  on $\sum_x m_x$ taxa will exactly fit this tree. If one then defines $D(a,b)$ as the expectation of $D_T(A_1,B_1)$ for each $a,b\in \mathcal X$, or as the expectation of the average of $D_T(A_i,B_j)$ over $1\le i\le m_a$ and $1\le j\le m_b$, the expected dissimilarity on $\mathcal X$ is the same, as the $A_i$ lineages for various $i$ are exchangeable under the coalescent model. Since this expected dissimilarity must exactly fit the unrooted topology relating only the $X_1$ for $x\in \mathcal X$, it thus fits the unrooted topology relating the taxa in $\mathcal X$. Thus either retaining only one sample per taxon, or averaging over the lineages sampled from each taxon will lead to consistent inference. Since data sets have only a finite number of gene trees, by averaging the empirical $D_T(A_i,B_j)$ one would hope to improve one's estimate of the expected value, so we choose to do so. Moreover, one could obtain the same dissimilarity by averaging over samples for each gene tree $T$ individually, creating a USTAR dissimilarity matrix for $\mathcal X$ from one tree at a time, and then averaging over these.

Now suppose we specify a finite number of multisample schemes, as well as probabilities of using each one for any gene. Given a data set of gene trees obtained from such an approach,  as described in the last paragraph one could apply a USTAR method averaged over multiple lineage samples to each subcollection of trees with the same multisample scheme.
But since the dissimilarity for each such subcollection in expectation approaches one fitting the species tree as the number of gene trees increases, then any weighted average of them over the multisample schemes does as well. This is a consequence of the dissimilarity arising from each subcollection satisfying the same four-point condition equality and inequalities, so a convex linear combination
of them does also.  Thus with multisample schemes $(m_{x,s})_{x\in \mathcal X}$ for $1\le s  \le S$, and any non-negative weighting constants $\alpha_s$, if we define an empirical dissimilarity as
\begin{equation}\hat D(a,b)=\sum_{1\le s\le S}\alpha_s \sum_ {T \text{ displaying }\atop(m_{x,s})} \frac 1{m_{a,s}m_{b,s}} \sum_{1\le i\le m_{a,s}\atop {1\le j\le n_{b,s}}} D_{T}(A_i,B_j)\label{eq:multi1}\end{equation}
then  we will have consistent inference provided the number of gene trees for each scheme in the sum all go to infinity.
Choosing $\alpha_s=1/|\mathcal T|$ where $\mathcal T$ is the collection of gene trees yields our suggested formula
\begin{equation}\hat D(a,b)=\frac 1{|\mathcal T|} \sum_ {T\in \mathcal T } \frac 1{m_a(T)m_b(T)} \sum_{1\le i\le m_a(T)\atop {1\le j\le m_b(T)}} D_{T}(A_i,B_j).\label{eq:multi2}\end{equation}

Note that the formula  \eqref{eq:multi1} can not be specialized to give formula \eqref{eq:Liumulti}. Taking $\alpha_s=m_{a,s}m_{b.s}/\sum_s m_{a,s}m_{b,s}$ does make them agree for the single comparison of $a$ and $b$, but will not for other pairs of taxa (unless $m_{x,s}$ is independent of $x$).

The essential difference between the formulas \eqref{eq:multi2} and \eqref{eq:Liumulti} is how the product ${m_a(T) m_b(T)}$ appears in them. In formula
\eqref{eq:Liumulti} all $D_T(A_i,B_j)$ are treated on an equal basis, whether they come from the same locus and are therefore correlated, or from different loci and thus independent trials of the coalescent process. Formula 
\eqref{eq:multi2} can be viewed as first constructing an intertaxon distance matrix for each locus by averaging pairwise distances over choices of alleles, and then 
averaging these over loci, to create a final intertaxon distance matrix.

We emphasize that using the consistency of  formulas \eqref{eq:multi1} and \eqref{eq:multi2} to justify their use in applying USTAR to finite data sets hinges on an assumption that every multisample scheme that appears in a  data set appears many times. Particularly for data sets  assembled from several earlier studies, there may be little commonality in the sampling scheme from one gene to the next. Simulations are needed to explore whether our formulas behave well under such circumstances.

\medskip

Simulations in \cite{ASTRIDWarnow} testing the performance of USTAR methods did not explore multisample schemes at all. However, in that work a new variant of a USTAR method that allows for gene trees missing some taxa was studied --- in the notation above the $m_x(T)$ could be 1 or 0.
Although such USTAR methods were reported to perform well on simulated data under these circumstances, theoretical justification for the particular approach taken has yet to be developed. Moreover, one should be cautious that if the test simulations involve random deletion of taxa from gene trees, they may not be relevant to empirical data sets in which taxa are missing in more patterned ways.

\section*{Acknowledgements}  
This work was begun while EA and JR were Short-term Visitors and JD was a Sabbatical Fellow at the National Institute for
Mathematical and Biological Synthesis, an institute sponsored by the National Science Foundation, the
U.S. Department of Homeland Security, and the U.S. Department of Agriculture through NSF Award
\#EF-0832858, with additional support from the University of Tennessee, Knoxville.
It was further supported by the National Institutes of Health grant R01 GM117590, awarded under the  Joint DMS/NIGMS Initiative to Support Research at the Interface of the Biological and Mathematical Sciences.

\bibliographystyle{IEEEtran}

\bibliography{USTAR}

\end{document}